\let\accentvec\vec

\documentclass{llncs}
\usepackage[utf8]{inputenc}
\usepackage{microtype}
\usepackage{srcltx}

\let\vec\accentvec
\usepackage{amsmath}
\usepackage{amssymb}
\usepackage{amsfonts}
\usepackage[textsize=tiny]{todonotes}

\spnewtheorem{algorithm}{Algorithm}{\bfseries}{\rmfamily}

\let\accentvec\vec

\let\vec\accentvec

\newcommand{\HH}{\mathcal{H}}

\newcommand{\ZZ}{\mathcal{Z}}
\newcommand{\ex}{\textnormal{\textrm{ex}}}
\newcommand{\E}{\textrm{E}}

\def\CS#1{\ensuremath{\textnormal{\textsf{CS}}^{(#1)}}}
\def\LCY#1{\ensuremath{\textnormal{\textsf{LCY}}^{(#1)}}}
\def\Lcyc{\ensuremath{\textnormal{\textsf{LCY}}}}

\def\IIN{\mathbb{N}}

\newcommand{\paren}[1]{{\left({#1}\right)}}

\newcommand{\abs}[1]{{\left\lvert{#1}\right\rvert}}
\def\LL{\ensuremath{\textnormal{\textsf{LL}}}}

\newcommand{\rmTheta}{\mathrm{\Theta}}

\title{Explicit and Efficient Hash Families
Suffice for Cuckoo Hashing with a Stash}

\author{Martin Aum\"{u}ller\inst{1}, Martin
Dietzfelbinger\inst{1}$^,$\thanks{Research supported by
DFG grant DI~412/10-2.}, and Philipp Woelfel\inst{2}$^,$\thanks{Research
supported by a Discovery Grant from the National Sciences and
Research Council of Canada (NSERC).}}
\institute{Fakult{\"a}t f{\"u}r Informatik und Automatisierung, Technische
Universit{\"a}t
Ilmenau,\\ 98694 Ilmenau, Germany\\ \email{martin.aumueller@tu-ilmenau.de},
\email{martin.dietzfelbinger@tu-ilmenau.de}\and Department of Computer Science,
University of Calgary,\\ Calgary, Alberta T2N 1N4, Canada\\
\email{woelfel@cpsc.ucalgary.ca}}

\begin{document}
\maketitle

\begin{abstract}
It is shown that for cuckoo hashing with a stash as proposed by Kirsch,
Mitzenmacher, and Wieder (2008) 
families of very simple hash functions can be used,
maintaining the favorable performance guarantees:
with stash size $s$ the probability of a rehash is $O(1/n^{s+1})$, 
and the evaluation time is $O(s)$. 
Instead of the full randomness
needed for the analysis of Kirsch \emph{et al.} and of Kutzelnigg (2010) 
(resp. $\rmTheta(\log n)$-wise independence for standard cuckoo hashing)
the new approach even works with 2-wise independent hash families as building blocks.  
Both construction and analysis build upon the work of Dietzfelbinger and Woelfel (2003).
The analysis, which can also be applied to the fully random case, 
utilizes a graph counting argument and is much simpler than previous proofs. 
As a byproduct, an algorithm for simulating uniform hashing is obtained. 
While it requires about twice as much space as the most space efficient solutions,  
it is attractive because of its simple and direct structure.

\end{abstract}

\section{Introduction}\label{sec:intro}
Cuckoo hashing as proposed by Pagh and Rodler~\cite{cuckoo_hashing_pagh} 
is a popular implementation of a dictionary with guaranteed constant lookup time. 
To store a set $S$ of $n$ keys from a universe $U$ (i.e., a finite set), cuckoo hashing utilizes two
hash functions, $h_1,h_2:U\to[m]$, where $m=(1+\varepsilon)n$, $\varepsilon>0$.
Each key $x\in S$ is stored in one of two hash tables of size $m$; 
either in the first table at location $h_1(x)$ or in the second one at location $h_2(x)$.
The pair $h_1,h_2$ might not be suitable to accommodate $S$ in these two tables. 
In this case, a \emph{rehash} operation is necessary, which chooses a new pair $h_1,h_2$ and inserts all keys anew.

In their ESA 2008 paper~\cite{stash}, Kirsch, Mitzenmacher, and Wieder deplored the order of magnitude of the probability of a rehash, which is as large as ${\rmTheta}(1/n)$.
They proposed adding a \emph{stash}, 
an additional segment of storage that can hold up to $s$ keys for some (constant) parameter $s$,
and showed that this change
reduces the rehash probability
to ${\rmTheta}(1/n^{s+1})$.
However, the analysis of Kirsch \emph{et al.} requires the hash functions to be
fully random.
In the journal version~\cite{stash:journal:09} 
Kirsch \emph{et al.} posed
 ``proving the above bounds for explicit hash families
that can be represented, sampled, and evaluated efficiently'' as an open problem.

\paragraph{\textnormal{\textbf{Our contribution.}}}\label{subsec:our:contribution}
In this paper we generalize a hash family construction proposed by Dietzfelbinger and Woelfel \cite{DW2003a} 
and show that the resulting hash functions 
have random properties strong enough to preserve the qualities of cuckoo hashing with a stash.
The proof involves a new and simpler analysis of this hashing scheme,
which also works in the fully random case.
The hash functions we propose have a very simple structure:
they combine functions from $O(1)$-wise independent families%
\footnote{$\kappa$-wise independent families of hash functions are defined 
in Section~\ref{sec:basics}.} 
with a few tables of size $n^{1-\rmTheta(1)}$ with random entries from 
$[m]=\{0,\ldots,m-1\}$. 
An attractive version of the construction for stash capacity $s$
has the following performance characteristics:
the description of a hash function pair $(h_1,h_2)$ 
consists of a table with $\sqrt{n}$ entries from $[m]^2$
and $2s+6$ functions from $2$-wise independent classes. 
To evaluate $h_1(x)$ and $h_2(x)$ for $x \in U$, 
we must evaluate these $2s+6$ functions, 
read $2s+4$ table entries, and carry out $4s+8$ additions modulo $m$.
Our main result implies 
for these hash functions and 
for any set $S\subseteq U$ of $n$ keys that with probability $1-O(1/n^{s+1})$
$S$ can be accommodated according to the cuckoo hashing rules.

 In addition, we present a simple data structure for simulating a uniform hash function on $S$ with range $R$, 
using our hash class and essentially a table with
$2(1+\varepsilon)n$ random elements from $R$.

\paragraph{\textnormal{\textbf{Cuckoo hashing with a stash and weak hash functions.}}}\label{subsec:cuckoo:stash}
In~\cite{stash:journal:09,Kutzelnigg10}
it was noticed that for the analysis of cuckoo hashing with a stash of size $s$ 
the properties of the so-called \emph{cuckoo graph} $G(S,h_1,h_2)$ are central.
Assume a set $S$ and hash functions $h_1$ and $h_2$ with range $[m]$ are given. 
The associated cuckoo graph $G(S,h_1,h_2)$ is the bipartite multigraph 
whose two node sets are copies of $[m]$ and whose edge set contains the $n$ pairs
$(h_1(x),h_2(x))$, for $x\in S$.
It is known that a single parameter of $G=G(S,h_1,h_2)$
determines whether a stash of size $s$ 
is sufficient to store $S$ using $(h_1,h_2)$, namely 
the \emph{excess} $\text{ex}(G)$, 
which is defined as the minimum number of edges
one has to remove from $G$ so that all connected components 
of the remaining graph are acyclic or unicyclic.
\begin{lemma}[\cite{stash:journal:09}]\label{lemma:stash:excess}
The keys from $S$ can be stored in the two tables and a stash of size $s$ using $(h_1,h_2)$ if and only
if \emph{$\text{ex}(G(S,h_1,h_2))\le s$}. 
\end{lemma}
For the convenience of the reader, a proof is given in Appendix~\ref{app:sec:insertions}, 
along with a discussion of insertion procedures, which is omitted in the main text.

Kirsch \emph{et al.} \cite{stash:journal:09} showed that with probability $1-O(1/n^{s+1})$ a random bipartite
graph with $2m=2(1+\varepsilon)n$ nodes and $n$ edges has excess at most $s$.
Their proof uses sophisticated tools such as Poissonization and Markov chain coupling. 
This result generalizes the analysis of standard cuckoo hashing~\cite{cuckoo_hashing_pagh}
with no stash, in which the rehash probability is $\rmTheta(1/n)$. 
Kutzelnigg~\cite{Kutzelnigg10} refined the analysis of~\cite{stash:journal:09}  
in order to determine the constant factor in the asymptotic bound of the rehash probability.
His proof uses generating functions and differential recurrence equations.
Both approaches inherently require 
that the hash functions $h_1$ and $h_2$ used in the algorithm are fully random. 

Recently, P\v{a}tra\c{s}cu and Thorup~\cite{patrascu11charhash} showed that simple tabulation hash functions
are sufficient for running cuckoo hashing, with a rehash probability of $\rmTheta(1/n^{1/3})$, which is tight.
Unfortunately, for these hash functions the rehash probability cannot be improved by using a stash.

Our main contribution is a new analysis that shows
that explicit and efficient hash families are sufficient to obtain the $O(1/n^{s+1})$ bound on the rehash probability. 
We build upon the work of Dietzfelbinger and Woelfel~\cite{DW2003a}.
For standard cuckoo hashing, they proposed hash functions of the form
$h_i(x)=\left(f_i(x) + z^{(i)}[g(x)]\right)\bmod m\text{, for }x\in U$,
for $i \in \{1,2\}$, where $f_i$ and $g$ are from $2k$-wise independent
classes with range $[m]$ and $[\ell]$, resp.,
and $z^{(1)}, z^{(2)}\in [m]^\ell$ are random vectors.
They showed that with such hash functions the rehash probability is $O(1/n + n/\ell^k)$.
Their proof has parts (i) and (ii).
Part (i) already appeared in~\cite{devroye} and~\cite{cuckoo_hashing_pagh}:
The rehash probability is bounded by 
the sum, taken over all minimal excess-$1$ graphs $H$ of different sizes 
and all subsets $T$ of $S$, of the probability that $G(T,h_1,h_2)$ is isomorphic to $H$.
In Sect.~\ref{sec:stash} of this paper we demonstrate that for $h_1$ and $h_2$ fully random
a similar counting approach also works for minimal excess-$(s+1)$ graphs,
whose presence in $G(S,h_1,h_2)$ determines whether a rehash is needed when a stash of size $s$ is used. 
As in~\cite{cuckoo_hashing_pagh}, this analysis also works for $O((s+1) \log n)$-wise independent families.

Part (ii) of the analysis in~\cite{DW2003a} is a little more subtle. 
It shows that for each key set $S$ of size $n$ there is a part $B^{\text{conn}}_S$
of the probability space given by $(h_1,h_2)$ such that $\Pr(B^{\text{conn}}_S)=O(n/\ell^{k})$
and in $\overline{B^{\text{conn}}_S}$ the hash functions act fully randomly on $T\subseteq S$
as long as $G(T,h_1,h_2)$ is connected.
In Sect.~\ref{sec:leafless} we show how this argument 
can be adapted to the situation with a stash, using 
subgraphs without leaves in place of the connected subgraphs.
Woelfel~\cite{asymmetric_balanced_allocation} already demonstrated by applying
functions in \cite{DW2003a} to balanced allocation 
that the approach has more general potential to it.

A comment on the ``full randomness assumption'' and work relating to it seems in order. 
It is often quoted as an empirical observation that weaker hash functions
like $\kappa$-wise independent families
will behave almost like random functions. 
Mitzenmacher and Vadhan~\cite{mitzenmacher_vadhan} showed that if the key set $S$
has a certain kind of entropy then $2$-wise independent hash functions will behave similar to fully
random ones. 
However, as demonstrated in \cite{DS09b}, there are situations 
where cuckoo hashing fails
for a standard $2$-wise independent family and even a random set $S$ (which is ``too dense'' in $U$).
The rather general ``split-and-share'' approach of \cite{DietzfelbingerR09} makes it possible to
justify the full randomness assumption for many situations involving hash functions, 
including cuckoo hashing and many of its variants. 
However, for practical application
this method is less attractive, since space consumption and failure probability are negatively affected
by splitting the key set into ``chunks'' and treating these separately.

\paragraph{\textnormal{\textbf{Simulating Uniform Hashing.}}}\label{subsec:sim:uniform}
Consider a universe $U$ of keys and a finite set $R$.  
By the term ``\emph{simulating uniform hashing for $U$ and $R$}'' 
we mean an algorithm that does the following. 
On input $n\in\IIN$, a randomized procedure sets up a data structure DS$_n$ 
that represents a hash function $h\colon U\to R$, 
which can then be evaluated efficiently for keys in $U$.
For each set $S\subseteq U$ of cardinality $n$ there is an event $B_S$
with the property that conditioned on $\overline{B_S}$ the values $h(x)$, $x\in S$, are fully random. 
The quality of the algorithm is determined by the space needed for
DS$_n$, the evaluation time for $h$, and the probability of the event $B_S$. 
It should be possible to evaluate $h$ in constant time.
The amount of entropy required for such an algorithm implies that at least 
$n\log|R|$ bits are needed to represent DS$_n$.

Pagh and Pagh~\cite{pagh_uniform} proposed a construction with $O(n)$ random words from $R$,
based on Siegel's functions~\cite{siegel04},
which have constant, but huge evaluation time. They also gave a general method to 
reduce the space to $(1+\varepsilon)n$, at the cost of an evaluation time
of $O(1/\varepsilon^2)$. 
In~\cite{DW2003a} a linear-space construction 
with tables of size $O(n)$ was given that contain (descriptions of) $O(1)$-wise independent hash functions.
The construction with the currently asymptotically best performance parameters,
$(1+\varepsilon)n$ words from $R$ and evaluation time $O(\log(1/\varepsilon))$,
as given in~\cite{DietzfelbingerR09}, is based on results of
Calkin~\cite{Calkin97} and
the ``split-and-share'' approach,
involving the same disadvantages as mentioned above. 

Our construction, to be described in Sect.~\ref{sec:uniform_hashing}, essentially results from the construction
in~\cite{pagh_uniform}
by replacing Siegel's functions with functions from our new class.
The data structure consists of a hash function pair $(h_1,h_2)$ 
from our hash class, a $O(1)$-wise independent hash function with range $R$, 
$O(s)$ small tables with entries from $R$,
and two tables of size $m=(1+\varepsilon)n$ each, filled with random elements from $R$.
The evaluation time of $h$ is $O(s)$, 
and for $S\subseteq U$, $|S|=n$, the event $B_S$ occurs with probability $O(1/n^{s+1})$.
The construction requires roughly twice as much space 
as the most space-efficient solutions~\cite{DietzfelbingerR09,pagh_uniform}. 
However, it seems to be a good compromise combining simplicity with moderate space consumption.

\section{Basics}\label{sec:basics}

Let $U$ (the ``universe'') be a finite set. A mapping from $U$ to $[r]$ is a \emph{hash function with range} $[r]$.
For an integer $\kappa\ge2$, a set $\HH$ of hash functions with range $[r]$ is called a 
$\kappa$\emph{-wise independent} hash family if for arbitrary distinct keys
$x_1,\ldots,x_\kappa\in U$ and for arbitrary $j_1,\ldots,j_\kappa\in [r]$ we have
$\Pr\nolimits_{h\in\HH}\bigl(h(x_1)=j_1 \wedge \ldots \wedge h(x_\kappa)=j_\kappa\bigr) =
{1}/{r^\kappa}$.
The classical $\kappa$-wise independent hash family construction is based on polynomials of degree $\kappa-1$
over a finite field~\cite{WegmanC79}.
More efficient hash function evaluation can be achieved with tabulation-based constructions \cite{DW2003a,ThorupZ04,ThorupZ10,KW2012a}.
Throughout this paper, $\HH^\kappa_r$ 
denotes an arbitrary $\kappa$-wise independent hash
family with domain $U$ and range $[r]$.

We combine $\kappa$-wise independent classes with lookups in
tables of size $\ell$ in order to obtain pairs of hash functions from $U$ to $[m]$:

\begin{definition}\label{def:family:Z}Let $c\ge1$ and $\kappa\ge2$. 
For integers $m$, $\ell\ge 1$, and given
$f_1,f_2\in  \HH^\kappa_m$, $g_1,\ldots,g_c \in \HH^\kappa_\ell$, 
and vectors $z^{(i)}_j\in[m]^\ell$, $1\le j \le c$, for $i\in\{1,2\}$,
let $(h_1,h_2) = (h_1,h_2)\langle f_1,f_2,g_1,\ldots,g_c,z_1^{(1)},\ldots,z_c^{(1)},z_1^{(2)},\ldots,z_c^{(2)}\rangle$, where
$$
{\textstyle h_i(x) = \left(f_i(x) +\sum_{1\le j \le c} z_j^{(i)}[g_j(x)]\right) \bmod m\text{, for }x\in U, i\in\{1,2\}.}
$$
Let $\ZZ^{\kappa,c}_{\ell,m}$ be the family of all these pairs $(h_1,h_2)$ of hash functions.
\end{definition}
While this is not reflected in the notation, we consider $(h_1,h_2)$ as a
structure from which 
the components $g_1,\ldots,g_c$ and $f_i,z^{(i)}_1,\ldots,z^{(i)}_c$, $i\in\{1,2\}$, can be read off again. 
It is family $\ZZ = \ZZ^{2k,c}_{\ell,m}$, for some $k \geq 1$, made into a probability space by the uniform distribution, that we
will study in the following.
We usually assume that $c$ and $k$ are fixed and that $m$ and $\ell$ are known.  

\subsection{Basic Facts}\label{subsec:basic:facts}
We start with some basic observations concerning 
the effects of compression properties in the ``$g$-part'' of $(h_1,h_2)$, extending similar statements in~\cite{DW2003a}.
\begin{definition}\label{def:T:bad}
For $T \subseteq U$, define the random variable 
$d_T$, the ``deficiency'' of $(h_1,h_2)$ with respect to $T$, by 
$d_T((h_1,h_2))=|T| - \max\{ k, |g_1(T)|,\ldots, |g_c(T)|\}$.
\emph{(}Note\emph{:} $d_T$ depends only on the $g_j$-components of $(h_1,h_2)$.\emph{)} 
Further, define\\
\makebox[2em][r]{\textnormal{(i)}} \mbox{\emph{$\text{bad}_T$} as the event that $d_T > k$\emph{;}}\\
\makebox[2em][r]{\textnormal{(ii)}} \emph{$\text{good}_T$}{} as \emph{$\overline{\text{bad}_T}$}, 
	           i.\,e., the event that $d_T \le k$\emph{;}\\
\makebox[2em][r]{\textnormal{(iii)}} \emph{$\text{crit}_T$} as the event that $d_T = k$.\\
Hash function pairs $(h_1,h_2)$ in these events are called 
\emph{``$T$-bad''}, \emph{``$T$-good''}, and \emph{``$T$-critical''}, resp.
\end{definition}
\begin{lemma}\label{lemma:random}
  Assume $k\ge1$ and $c\ge 1$. 
  For $T\subseteq U$, the following holds\emph{:}\\
\textnormal{(a)} \emph{$\Pr(\text{bad}_T\cup\text{crit}_T) \le \bigl(\abs{T}^{2k}/\ell^k\bigr)^c$}.\\
\textnormal{(b)} Conditioned on \emph{$\text{good}_T$} \emph{(}or on \emph{$\text{crit}_T$}\emph{)},
  the pairs $(h_1(x),h_2(x))$, $x\in T$, are \linebreak 
  \phantom{\textnormal{(b)}}distributed uniformly and independently in $[r]^2$.
\end{lemma}
\begin{proof}
(a) Assume $\abs{T}\ge 2k$ (otherwise the events $\text{bad}_T$ and $\text{crit}_T$ cannot occur).
Since $g_1,\dots,g_c$ are
independent, it suffices to show that for a function $g$ chosen randomly from $\HH^{2k}_\ell$ we have
$\Pr( |T| - |g(T)| \geq k ) \le \abs{T}^{2k}/\ell^k$.

We first argue that if $|T| - |g(T)| \ge k$ then there is
a subset $T'$ of $T$
with
$\abs{T'}=2k$
    and
$\abs{g(T')}\leq k$.
Initialize $T'$ as $T$. Repeat the following as long as $|T'| > 2k$: (i) if
there exists a key $x \in T'$ such that $g(x) \neq g(y)$ for all $y \in T' \setminus \{x\}$,
remove $x$ from $T'$; (ii) otherwise, remove any key. 
Clearly, this process terminates with $\abs{T'}=2k$. 
It also maintains the invariant $\abs{T'}-\abs{g(T')}\geq k$: 
In case (i) $\abs{T'} - \abs{g(T')}$ remains unchanged.
In case (ii) before the key is removed from $T'$ we have
$|g(T')| \leq |T'|/2$ and thus $|T'| - |g(T')| \geq |T'|/2 > k$.

Now fix a subset 
$T'$ of $T$ of size $2k$ that satisfies $\abs{g(T')}\leq k$.
The preimages $g^{-1}(u)$, $u\in g(T')$, partition $T'$ into
$k'$ classes, $k'\le k$, such that $g$ is constant on each class.  
Since $g$ is chosen from a $2k$-wise independent class, 
the probability that $g$ is constant on all classes of a given partition 
of $T'$ into classes $C_1,\ldots,C_{k'}$, with $k'\le k$, is exactly $\ell^{-(2k-k')} \le \ell^{-k}$.

Finally, we bound 
$\Pr(|g(T)|\le |T|-k)$.
There are $\binom{|T|}{2k}$ subsets $T'$ of $T$ of size $2k$.
Every partition of such a set $T'$ into $k'\leq k$ classes can be represented by a permutation of $T'$ with $k'$ cycles, 
where each cycle contains the elements from one class.
Hence, there are at most $(2k)!$ such partitions.
This yields:
\begin{equation}
\Pr(|T|-|g(T)|\geq k) \le \binom{|T|}{2k}\cdot (2k)!
\cdot\frac{1}{\ell^k} \le \frac{|T|^{2k}}{\ell^k}. 
\label{eq:600}
\end{equation}

\noindent(b) If $|T| \leq 2k$, then $h_1$ and $h_2$ are fully random on $T$ 
simply because $f_1$ and $f_2$ are $2k$-wise independent.
So suppose $|T| > 2k$. 
Fix an arbitrary $g$-part of $(h_1,h_2)$ so that $\text{good}_T$ occurs, 
i.e., $\max\{k, |g_1(T)|,\ldots, |g_c(T)|\} \geq \abs{T} - k$. 
Let $j_0 \in \{1,\dots,c\}$ be such that $|g_{j_0}(T)| \geq |T| - k$. 
Arbitrarily fix all values in the tables $z_{j}^{(i)}$ with $j\neq j_0$
and $i \in \{1,2\}$. 
Let $T^\ast$ be
the set of keys in $T$ colliding with other keys in $T$ under $g_{j_0}$. 
Then $|T^\ast| \leq 2k$.
Choose the values $z_{j_0}^{(i)}[g_{j_0}(x)]$ for all $x \in T^\ast$ and $i\in \{1,2\}$ at random.
Furthermore, choose $f_1$ and $f_2$ at random from the $2k$-wise independent
family $\HH^{2k}_r$. 
This determines $h_1(x)$ and $h_2(x)$, $x \in T^\ast$, as fully random values.
Furthermore, the function $g_{j_0}$ maps the keys $x \in T - T^\ast$ to
distinct entries of the vectors $z_{j_0}^{(i)}$ that were not fixed before.
Thus, the hash function values $h_1(x), h_2(x)$, $x \in T - T^\ast$, are 
distributed fully randomly as well and are independent of those with $x\in T^\ast$.
\qed \end{proof}
%

\section{Graph Properties and Basic Setup}\label{sec:graph:properties}

For $m \in \IIN$ let $\mathcal{G}_m$ denote the set of all
bipartite (multi-)graphs with vertex set $[m]$ on each side of the bipartition. 
A set $\mathcal{A} \subseteq \mathcal{G}_m$ is called a \emph{graph property}.
For example, $\mathcal{A}$ could be the set of graphs in $\mathcal{G}_m$ that have excess larger than $s$.
For a graph property $\mathcal{A} \subseteq \mathcal{G}_m$ and $T\subseteq U$, let
$\mathcal{A}_T$ denote the event
that $G(T,h_1,h_2)$ has property $\mathcal{A}$ (i.\,e., that $G(T,h_1,h_2)\in\mathcal{A}$).
In the following, our main objective is to bound the probability $\Pr(\exists T \subseteq S\colon \mathcal{A}_T)$
for graph properties $\mathcal{A}$ which are important for our analysis.

For the next lemma we need the following definitions. For $S \subseteq U$ and a graph property $\mathcal{A}$ let $B^{\mathcal{A}}_{S}
\subseteq \ZZ$  be the event 
 $\exists T \subseteq S \colon \mathcal{A}_T \cap \text{bad}_T$ (see Def.~\ref{def:T:bad}). Considering fully random hash functions $(h_1^*,h_2^*)$ for a moment, let
 $p^{\mathcal{A}}_T=\Pr(G(T,h_1^*,h_2^*)\in\mathcal{A})$.
\begin{lemma}\label{lem:good:bad} For  an arbitrary graph property $\mathcal{A}$
we have
\begin{equation}
\Pr(\exists T \subseteq S\colon \mathcal{A}_T) \le \Pr(B^{\mathcal{A}}_{S}) + \sum_{T\subseteq S}p^{\mathcal{A}}_T. 
\label{eq:1000}
\end{equation}
\end{lemma}
\begin{proof}\quad 
 $\Pr(\exists T \subseteq S\colon \mathcal{A}_T)
  \leq \Pr(B^{\mathcal{A}}_{S}) + \Pr((\exists T \subseteq S: \mathcal{A}_T)
\cap \overline{B^{\mathcal{A}}_{S}})$, and
\begin{align*}
\sum_{T \subseteq S}\Pr( \mathcal{A}_T
\cap \overline{B^{\mathcal{A}}_{S}}) \underset{\text{(i)}}{\leq} \sum_{T \subseteq
S}\Pr( \mathcal{A}_T \cap \text{good}_T ) \leq \sum_{T \subseteq S}\Pr(
\mathcal{A}_T \mid \text{good}_T) \underset{\text{(ii)}}{=} \sum_{T \subseteq
S}p^{\mathcal{A}}_T,
\end{align*}
where (i) holds by the definition of $B^{\mathcal{A}}_{S}$,
and (ii) holds by Lemma~\ref{lemma:random}(b).
\qed
\end{proof}
This lemma encapsulates our overall strategy for bounding $\Pr(\exists
T\subseteq S\colon \mathcal{A}_T)$. The second summand in (\ref{eq:1000}) can be bounded assuming full randomness.
The task of bounding the first summand is tackled separately, in Section~\ref{sec:leafless}.

\section{A Bound for Leaf\/less Graphs}\label{sec:leafless}

The following observation, which is immediate from the definitions,
will be helpful in applying Lemma~\ref{lem:good:bad}.
\begin{lemma}\label{lem:subproperty}
Let $m \in \mathbb{N}$, and let $\mathcal{A} \subseteq \mathcal{A'} \subseteq
\mathcal{G}_m$. Then $\Pr(B^{\mathcal{A}}_{S}) \leq 
\Pr(B^{\mathcal{A}'}_{S})$.\qed 
\end{lemma}
We define a graph property to be used in the role
of $\mathcal{A}'$ in applications of Lemma~\ref{lem:subproperty}.  
A node with degree 1 in a graph is called a \emph{leaf};
an edge incident with a leaf is called a \emph{leaf edge}.
An edge is called a \emph{cycle edge} if removing it does not disconnect any two nodes. 
  A graph is called \emph{leafless} if it has no leaves. 
 Let $\LL \subseteq \mathcal{G}_m$ be the set of all leaf\/less graphs. The  
\emph{$2$-core} of a graph is its (unique) maximum leaf\/less subgraph.
The purpose of the present section is to prove a bound on $\Pr(B^{\LL}_{S})$.
\begin{lemma}\label{lem:leafless}
Let $\varepsilon>0$, let $S\subseteq U$ with $|S|=n$, and let $m=(1+\varepsilon)n$.
Assume $(h_1,h_2)$ is chosen at random from $\ZZ=\ZZ^{2k,c}_{\ell,m}$. 
Then $\Pr(B^{\LL}_{S}) = O\paren{{n}/{\ell^{ck}}}$. 
\end{lemma}
We recall a standard notion from graph theory (already used in~\cite{DW2003a}; \emph{cf}. App.~\ref{subsec:excess}):
The \emph{cyclomatic number} $\gamma(G)$ of a graph $G$ is the smallest number of edges 
one has to remove from $G$ to obtain a graph with no cycles.
Also, let $\zeta(G)$ denote the number of connected components of $G$ 
(ignoring isolated points). 
\begin{lemma}\label{lem:num_graphs}
Let $N(t,\ell, \gamma, \zeta)$ be the number of non-isomorphic 
\emph{(}multi-\emph{)}graphs with $\zeta$ connected components and
cyclomatic number $\gamma$ that have $t$ edges, $\ell$ of which are leaf
edges.
Then $N(t,\ell,\gamma,\zeta) = t^{O(\ell + \gamma + \zeta)}$.
\end{lemma}
\begin{proof}
In \cite[Lemma 2]{DW2003a} it is shown that $N(t,\ell,\gamma,1) = t^{O(\ell+\gamma)}$. 
Now note that each graph $G$ with cyclomatic number $\gamma$, $\zeta$ connected components, 
$t-\ell$ non-leaf edges, and $\ell$ leaf edges 
can be obtained from some connected graph $G'$ with cyclomatic number $\gamma$,
$t-\ell+\zeta-1$ non-leaf edges, and $\ell$ leaf edges by removing $\zeta-1$ non-leaf, non-cycle edges.
There are no more than $(t-\ell+\zeta-1)^{\zeta-1}$ ways for choosing the edges to be removed. 
This implies, using  \cite[Lemma 2]{DW2003a}: 
    \begin{align*} 
    N(t,\ell,\gamma,\zeta) & \le   N(t+\zeta-1,\ell,\gamma,1) \cdot (t-\ell+\zeta-1)^{\zeta-1}\\ 
    &\le (t+\zeta)^{O(\ell+\gamma)} \cdot (t+\zeta)^{\zeta}
    = (t+\zeta)^{O(\ell + \gamma + \zeta)} = t^{O(\ell + \gamma + \zeta)}.\tag*{\qed}
    \end{align*}
\end{proof}
We shall need more auxiliary graph properties:
A graph from $\mathcal{G}_m$
belongs to $\Lcyc$ if at most one connected component contains leaves 
(the \emph{leaf component}); for 
$K\ge1$ it belongs to $\LCY{K}$ if it has the following four properties:
\begin{enumerate} 
\item at most one connected component of $G$ contains leaves (i.\,e., $\LCY{K}\subseteq \Lcyc$);
\item the number $\zeta(G)$ of connected components is bounded by $K$;
\item if present, the leaf component of $G$ contains at most $K$ leaf and cycle edges;
\item the cyclomatic number $\gamma(G)$ is bounded by $K$.
\end{enumerate}
\begin{lemma}\label{lem:02}
If $T \subseteq U$ and $(h_1,h_2)$ is from $\ZZ$ such that $G(T,h_1,h_2) \in \LL$ and
$(h_1,h_2)$ is $T$-bad, then there exists a subset $T'$ of $T$ such that
$G(T',h_1,h_2) \in \LCY{4ck}$ and $(h_1,h_2)$ is $T'$-critical.
\end{lemma}
\begin{proof}
Fix $T$ and $(h_1,h_2)$ as in the assumption. Initialize $T'$ as $T$. 
We will remove (``peel'') edges from $G(T',h_1,h_2)$ in four stages.
Of course, by ``removing edge $(h_1(x),h_2(x))$ from $G(T',h_1,h_2)$'' we mean
removing $x$ from $T'$.

\emph{Stage} 1: 
Initially, we have $d_{T'}((h_1,h_2)) > k$. 
Repeat the following step:
If $G(T',h_1,h_2)$ contains a leaf,
remove a leaf edge from it, otherwise remove a cycle edge.
Clearly, such steps maintain the property that $G(T',h_1,h_2)$ belongs to $\Lcyc$. 
Since $d_{T'}((h_1,h_2))$ can decrease by at most 1 when an edge is removed,
we finally reach a situation where $d_{T'}((h_1,h_2))=k$,
i.\,e., $(h_1,h_2)$ is $T'$-critical.
Then $G(T',h_1,h_2)$ satisfies Property 1 from the definition of $\LCY{4ck}$.

To prepare for the next stages, 
we define a set $T^\ast\subseteq T'$
with $2k \leq |T^\ast| \leq 2ck$,
capturing keys that have to be ``protected'' during the following stages
to maintain criticality of $T'$. 
If $|T'| = 2k$, we simply let $T^\ast = T'$. Then $(h_1,h_2)$ is $T^\ast$-critical. 
If $|T'| > 2k$, a little more work is needed.  
By the definition of $d_T((h_1,h_2))$ we have $|T'|-\max\{|g_1(T')|,\ldots, |g_c(T')|\} = k$.
For each $j \in \{1,\dots,c\}$ Lemma~\ref{lemma:random}(a) gives us a set $T^\ast_j \subseteq T'$ such that 
$|T^\ast_j| = 2k$ and $|g_j(T^\ast_j)| \leq k$. 
Let $T^\ast := T^\ast_1\, \cup\, \ldots\, \cup\, T^\ast_c$. 
Clearly, $2k \leq |T^\ast| \leq 2ck$. Since $T_j^\ast\subseteq T^\ast$,
we have $|T^\ast| - |g_j(T^\ast)| \geq |T_j^\ast| - |g_j(T_j^\ast)| \geq k$, for
$1\le j \le c$, and hence $d_{T^\ast}((h_1,h_2)) \ge k$.
On the other hand we know that there exists some $j$ with $|T'| - |g_j(T')| = k$.
Since $T^\ast\subseteq T'$, we have $|T^\ast| - |g_j(T^\ast)| \le k$ for this $j$.
Altogether we get $d_{T^\ast}((h_1,h_2)) = k$,
which means that $(h_1,h_2)$ is $T^\ast$-critical also in this case.

Now we ``mark'' all edges of $G=G(T',h_1,h_2)$ whose keys belong to $T^\ast$.

\emph{Stage} 2: Remove all components of $G$ without marked edges. Afterwards there are
at most $2ck$ components left, and $G$ satisfies Property $2$.

\emph{Stage} 3: If $G$ has a leaf
component $C$, repeatedly remove unmarked leaf and cycle edges 
from $C$, while $C$ has such edges. The remaining leaf and cycle edges in $C$ are marked, 
and thus there number is at most $2ck$; Property 3 is satisfied. 

\emph{Stage} 4: If there is a leaf component $C$ with $z$ marked edges (where $z\le2ck$),
then $\gamma(C) \leq z - 1$. Now consider a leaf\/less component $C'$ with cyclomatic number $z$. 
We need the following claim, which is proved in Appendix~\ref{app:claim:proof}.
\begin{claim}\label{lem:cyclic_cyclo_peeling}
Every leaf\/less connected graph with $i$ marked edges has a leaf\/less connected subgraph with 
cyclomatic number $\le i {+} 1$ that contains all marked edges. 
\end{claim}
This claim gives us a leaf\/less subgraph $C''$ of $C'$ with $\gamma(C'') \leq z+1$ that contains all marked edges of $C'$. 
We remove from $G$ all vertices and edges of $C'$ that are not in $C''$.
Doing this for all leaf\/less components
yields the final key set $T'$ and the final graph $G=G(T',h_1,h_2)$. 
Summing contributions to the cyclomatic number of $G$
over all (at most $2ck$) connected components, we see that $\gamma(G) \leq 4ck$; Property 4 is satisfied. \qed
\end{proof}
What have we achieved? 
By Lemma~\ref{lem:02}, we can bound $\Pr(B^\LL_S)$ by just adding, over all 
$T'\subseteq S$, 
the probabilities $\Pr(\LCY{4ck}_{T'} \cap \text{crit}_{T'})$,
that means, the terms $\Pr(\LCY{4ck}_{T'} \mid \text{crit}_{T'})\cdot\Pr(\text{crit}_{T'})$.
Lemma~\ref{lemma:random}(a) takes care of the second factor.
By Lemma~\ref{lemma:random}(b), we may assume that $(h_1,h_2)$ acts fully random on $T'$ for the first factor.
The next lemma estimates this factor, using the notation from Section~\ref{sec:graph:properties}.
\begin{lemma}\label{lem:fully:random}
 Let $T \subseteq U, |T|=t$, and $c,k \geq 1$. Then
 $p^{\LCY{4ck}}_T \!\! \leq t! \cdot t^{O(1)}/m^{t - 1}.$
\end{lemma}
\begin{proof}
By Lemma~\ref{lem:num_graphs}, there are at most $t^{O(ck)} =
t^{O(1)}$ ways to choose a bipartite graph $G$ in $\LCY{4ck}$ with $t$ edges.
Graph $G$ cannot have more than $t+1$ nodes, since cyclic
components have at most as many nodes as edges, and in the single leaf component, if 
present, the number of nodes is at most one bigger than the number of edges.  
In each component of $G$, there are two ways to assign the vertices to the two sides of the
bipartition. After such an assignment is fixed, there are at most $m^{t+1}$ ways
to label the vertices with elements of $[m]$, and there are $t!$ ways
to label the edges of $G$ with the keys in $T$. Assume now such labels have been chosen for $G$.
Draw $t$ edges $(h_1^\ast(x),h_2^\ast(x))$ from $[m]^2$ uniformly at random.
The probability that they exactly fit the labeling of nodes and edges of $G$ is $1/m^{2t}$.  
Thus, $p^{\LCY{4ck}}_T \leq m^{t+1}\cdot t! \cdot t^{O(1)}/m^{2t} = t! \cdot
t^{O(1)}/m^{t-1}$. \qed
\end{proof}
We can now finally prove Lemma~\ref{lem:leafless}, the main lemma of this section.
\begin{proof}[of Lemma~\ref{lem:leafless}]
By Lemma~\ref{lem:02}, and using the union bound, we get
\begin{align*}
\Pr(B^\LL_S)&= \Pr(\exists T \subseteq S: \LL_T \cap \text{bad}_T) \leq
\Pr(\exists T' \subseteq S: \LCY{4ck}_{T'} \cap \text{crit}_{T'})\\
&\leq \sum_{T' \subseteq S} \Pr( \LCY{4ck}_{T'} \mid \text{crit}_{T'})
\cdot
\Pr(\text{crit}_{T'}) =: \rho_S.
\end{align*}
By Lemma~\ref{lemma:random}(b),
given the event that $(h_1,h_2)$ is $T'$-critical, 
$(h_1,h_2)$ acts fully random on $T'$. 
Using Lemma~\ref{lem:fully:random} and
Lemma~\ref{lemma:random}(a), this yields:
$$\Pr(\LCY{4ck}_{T'} \mid \text{crit}_{T'}) \cdot \Pr(\text{crit}_{T'}) \leq (|T'|! \cdot
|T'|^{O(1)}/m^{|T'| - 1}) \cdot (|T'|^2/\ell)^{ck}.$$
Summing up, collecting sets $T'$ of equal size together, and using that $ck$ is
constant, we obtain
\begin{equation}
 \rho_S \leq \!\!\!\sum_{2k\le t\le n}\! \binom{n}{t} \cdot
\frac{t! \cdot t^{O(1)}}{m^{t-1}} \cdot \left(\frac{t^2}{\ell}\right)^{ck}
\!\leq   \frac{n}{\ell^{ck}} \cdot \!\sum_{2k\le t\le n}\! \frac{t^{O(1)} }{ (1
+ \varepsilon)^{t-1} } = O\paren{\frac{n}{\ell^{ck}}}.\tag*{\qed}
\label{eq:980}
\end{equation}
\end{proof}

\section{Cuckoo Hashing With a Stash}\label{sec:stash}

In this section we prove the desired bound on the rehash probability of cuckoo hashing
with a stash when functions from $\ZZ$ are used. 
We focus on the question whether the pair $(h_1,h_2)$ allows storing key set $S$
in the two tables with a stash of size $s$.
In view of Lemma~\ref{lemma:stash:excess}, we identify minimal graphs with
excess $s+1$. 
\begin{definition}
     An \emph{excess-$(s+1)$ core graph} is a leafless graph $G$ with excess
exactly $s + 1$ in which
     all connected components have at least two cycles.
     By $\CS{s+1}$ we denote the set of all excess-$(s+1)$ core graphs in
$\mathcal{G}_m$.
\end{definition}
\begin{lemma}\label{lem:excess_core_structure}
    Let $G=G(S,h_1,h_2)$ be a cuckoo graph with $\ex(G)\geq s + 1$.
    Then $G$ contains an excess-$(s+1)$ core graph as a subgraph. 
\end{lemma}
\begin{proof}
We repeatedly remove edges from $G$.
First we remove cycle edges until the excess is exactly $s+1$.
Then we remove components that are trees or unicyclic.
Finally we  remove leaf edges one by one until the remaining 
graph is leaf\/less. 
\qed
\end{proof}
We are now ready to state our main theorem.
\begin{theorem}\label{thm:r_excess_prob_zz}\samepage
    Let $\varepsilon > 0$ and $0 < \delta <1$, let  $s\ge0$ and $k\ge1$ be
given. Assume $c\ge  (s+2)/(\delta k) $.  
    For $n\ge 1$ consider $m\ge(1+\varepsilon)n$ and $\ell=n^\delta$.  
    Let $S \subseteq U$ with $|S|=n$.
    Then for $(h_1,h_2)$ chosen at random from $\ZZ=\ZZ^{2k,c}_{\ell,m}$ the
following holds:
    $$\Pr(\ex(G(S,h_1,h_2)) \geq s+1) =
O(1/n^{s+1}).$$
\end{theorem}   
\begin{proof}
By Lemma~\ref{lem:good:bad} and Lemma~\ref{lem:excess_core_structure}, 
the probability that the excess of $G(S,h_1,h_2)$ is at least $s+1$ 
is
\begin{equation}
\Pr(\exists T\subseteq S \colon \CS{s+1}_T) \leq \Pr(B^{\CS{s+1}}_S) + \sum_{T
\subseteq S} p^{\CS{s+1}}_T.
\label{eq:1500}
\end{equation}
Since $\CS{s+1}\subseteq \LL$, we can combine
Lemmas~\ref{lem:subproperty} and~\ref{lem:leafless} to obtain
$ \Pr(B^{\CS{s+1}}_S) \le \Pr(B^{\LL}_S) = O(n / \ell^{ck})$. 
It remains to bound the second summand in (\ref{eq:1500}).%
\footnote{We remark that the following
calculations also give an alternative, simpler proof of
\cite[Theorem 2.1]{stash} for the fully random case,
even if the effort needed to prove Lemma~\ref{lem:num_graphs}
and~\cite[Lemma 2]{DW2003a} is taken into account.} 
\begin{lemma}\label{lem:stash:fully:random}
${\sum_{T \subseteq S} p^{\CS{s+1}}_T = O(1/n^{s+1})}$.
\end{lemma}
\begin{proof}
We start by counting (unlabeled) excess-$(s+1)$ core graphs with $t$ edges. 
A connected component $C$ of such a graph $G$ with cyclomatic number $\gamma(C)$ (which is at least 2) 
contributes $\gamma(C)-1$ to the excess of $G$. 
This means that if $G$ has $\zeta=\zeta(G)$ components,
then $s+1 = \gamma(G) - \zeta$ and $\zeta \le s+1$, and hence $\gamma = \gamma(G) \le 2(s+1)$.
Using Lemma~\ref{lem:num_graphs}, there are at most 
$N(t,0,\gamma,\zeta)=t^{O(\gamma+\zeta)} = t^{O(s)}$ such graphs $G$. 
If from each component $C$ of such a graph $G$ we remove $\gamma(C)-1$ cycle edges, we get unicyclic components, 
which have as many nodes as edges. This implies that $G$ has $t-(s+1)$ nodes.

Now fix a bipartite (unlabeled) excess-$(s+1)$ core graph $G$ with $t$ edges and $\zeta$ components,
and let $T\subseteq U$ with $|T|=t$ be given. 
There are $2^\zeta \le 2^{s+1}$ ways of assigning the $t-s-1$ nodes to the two sides of the bipartition,
and then at most $m^{t-s-1}$ ways of assigning labels from $[m]$ to the nodes.
Thus, the number of bipartite graphs with property $\CS{s+1}$,
where each node is labeled with one side of the bipartition and an element of $[m]$,
and where the $t$ edges are labeled with distinct elements of $T$ is smaller than
$t!\cdot 2^{s+1} \cdot m^{t-s-1} \cdot t^{O(s)}$. 

Now if $G$ with such a labeling is fixed, and we choose $t$ edges from $[m]^2$ uniformly at random,
the probability that all edges $(h_1(x),h_2(x))$, $x\in T$, match the 
labeling is $1/m^{2t}$. For constant $s$, this yields the following bound: 
\begin{align*}
  \sum_{T \subseteq S} p^{\CS{s+1}}_T 
 &\leq \sum_{s+3 \le t \le n} \; \binom{n}{t}\frac{2^{s+1} \cdot n^{t - s - 1}
\cdot t! \cdot t^{O(s)}}{m^{2t}}
 \leq  \frac{2^{s+1}}{n^{s+1}} \cdot \sum_{s+3 \le t \le n} \frac{n^{t}
\cdot
t^{O(1)}}{m^{t}}\\
& = O\paren{\frac{1}{n^{s+1}}} \cdot \sum_{s+3 \le t \le n} \frac{t^{O(1)}}{(1+\varepsilon)^t} =
O\paren{\frac{1}{n^{s+1}}}.\tag*{\qed}
\end{align*}
\end{proof}
Since $\ell=n^\delta$ and $c\ge
(s+2)/(k\delta)$, we get
\begin{equation*}
 \Pr(\ex(G) \geq s + 1) = O(n / \ell^{ck}) + O(1/n^{s+1}) =
O(1/n^{s+1}).\tag*{\qed} 
\end{equation*}

\end{proof}

\section{Simulating Uniform Hashing}\label{sec:uniform_hashing}
In the following, let $R$ be the range of the hash function to construct, 
and assume that $(R,\oplus)$ is a commutative group. (We could use $R=[t]$ with
addition mod $t$.) 
\begin{theorem}\label{thm:uniform_hashing}
Let $n \geq 1, 0 < \delta < 1$, $\varepsilon > 0$, and $s\ge0$ be
given.
There exists a data structure \emph{DS$_n$} that allows
us to compute a function $h\colon U \rightarrow R$ such that\emph{:}\\
\makebox[2em][r]{\textnormal{(i) }}For each $S \subseteq U$ of size $n$ there is an event
$B_S$ of probability $O(1/n^{s+1})$ \linebreak 
  \phantom{\makebox[2em][r]{\textnormal{(i) }}}such that conditioned on $\overline{B_S}$ the function $h$ is distributed uniformly on $S$.\\
\makebox[2em][r]{\textnormal{(ii) }}For arbitrary $x \in U$, $h(x)$ can be evaluated in time
$O(s/\delta)$.\\
\makebox[2em][r]{\textnormal{(iii) }}\emph{DS$_n$} 
comprises $2(1+\varepsilon)n + O(sn^\delta)$ words from $R$ and $O(s)$ words from $U$. 
\end{theorem}
\begin{proof}
Let $k\geq 1$ and choose $c\ge (s+2)/(k\delta)$. Given $U$ and $n$, set up 
DS$_n$ as follows.
Let $m=(1+\varepsilon)n$ and $\ell=n^\delta$, and choose and store a hash function pair
$(h_1,h_2)$ from $\ZZ=\ZZ^{2k,c}_{\ell,m}$ (see Definition~\ref{def:family:Z}),
with component functions $g_1,\ldots,g_c$ from $\HH^{2k}_\ell$.
In addition, choose $2$ random vectors $t_1,t_2 \in R^m, c$ random vectors
$y_1,\dots,y_c \in R^\ell$, and choose $f$ at random from a $2k$-wise
independent
family of hash functions from $U$ to $R$. 

Using DS$_n$, the mapping $h\colon U\to R$ is defined as follows: 
$$
h(x) = t_1[h_1(x)] \oplus t_2[h_2(x)] \oplus f(x) \oplus y_1[g_1(x)] \oplus \ldots \oplus y_c[g_c(x)].
$$  
DS$_n$ satisfies (ii) and (iii) of
Theorem~\ref{thm:uniform_hashing}. We show that it satisfies (i) as well.

First, consider only the hash functions $(h_1,h_2)$ from $\ZZ$.
By Lemma~\ref{lem:leafless} we have $\Pr(B^{\LL}_{S}) =
O(n/\ell^{ck}) = O(1/n^{s+1})$.
Now fix $(h_1,h_2)\notin {B^{\LL}_{S}}$, which includes fixing the components $g_1,\ldots,g_c$.
Let $T\subseteq S$ be such that $G(T,h_1,h_2)$ is the 2-core of $G(S,h_1,h_2)$,
the maximal subgraph with minimum degree 2. Graph $G(T,h_1,h_2)$ is leaf\/less, 
and since $(h_1,h_2)\notin {B^{\LL}_{S}}$, we have that $(h_1,h_2)$ is $T$-good.
Now we note that the part $f(x) \oplus \bigoplus_{1\le j \le c}y_j[g_j(x)]$ of
$h(x)$
acts exactly as one of our hash functions $h_1$ and $h_2$ (see Definition~\ref{def:family:Z}(a)), 
so that arguing as in the proof of Lemma~\ref{lemma:random}
we see that $h(x)$ is fully random on $T$.

Now assume that $f$ and the entries in the tables $y_1,\ldots,y_c$ are fixed. 
It is not hard to show that the random entries in $t_1$ and $t_2$ alone make
sure 
that $h(x)$, $x\in S-T$, is fully random. 
(Such proofs were given in \cite{DW2003a} and \cite{pagh_uniform}.)  
\qed
\end{proof}

\subsection*{Concluding Remarks}
We presented a family of efficient hash functions and 
showed that it exhibits sufficiently strong random properties 
to run cuckoo hashing with a stash, 
preserving the favorable performance guarantees of this hashing scheme.
We also described a simple construction for simulating uniform hashing. 
We remark that the performance of our construction 
can be improved by using $2$-universal hash families\footnote{A family $\mathcal{H}$ of hash functions with range $R$ 
is $2$-universal if for each pair $x,y \in U$, $x \neq y$, 
and $h$ chosen at random from $\mathcal{H}$ we have $\Pr(h(x)=h(y)) \leq 2/|R|$.} 
(see, e.\,g., \cite{CarterW79,DietzfelbingerHKP97}) for the $g_j$-components. 
The proof of Lemma~\ref{lemma:random} can be adapted easily to these weaker families.
It remains open whether generalized cuckoo hashing 
\cite{FotakisPSS05,DietzfelbingerW07} can be run with efficient hash families.
\bibliographystyle{splncs03}
\bibliography{lit}

\appendix

\section{Excess, Stash Size, and Insertions}\label{app:sec:insertions}
In this supplementary section, provided for the convenience of the reader,
we clarify the connection between stash size needed and the excess $\text{ex}(G(S,h_1,h_2))$ 
of the cuckoo graph $G(S,h_1,h_2)$ as well as the role of insertion procedures.
In particular, we prove Lemma~\ref{lemma:stash:excess}.
The central statements of this section 
can also be found in~\cite{stash:journal:09,Kutzelnigg10}.

\subsection{The Excess of a Graph}\label{subsec:excess}
For $G$ a graph, $\zeta(G)$ denotes the number of connected components of $G$.
The cyclomatic number $\gamma(G)$, technically defined as ``the dimension 
of the cycle space of $G$'', can be characterized by the following basic formula~\cite{Diestel}:
\begin{equation}
\gamma(G) = m - n + \zeta(G),
\label{eq:500}
\end{equation}
for $n$ the number of nodes and $m$ the number of edges of $G$.
Note that acyclic graphs are characterized by the equation $n=m+\zeta(G)$ and hence 
by the equation $\gamma(G)=0$.
Using (\ref{eq:500}), two helpful ways of viewing $\gamma(G)$ are easy to prove.

\begin{lemma}\label{lem:cyclo:one}
\begin{itemize}
	\item[\textnormal{(a)}] If we remove edges from $G$ sequentially, in an arbitrary order,
	and the resulting graph is acyclic, then $\gamma(G)$ is the number of cycle edges removed. 
	\item[\textnormal{(b)}] $\gamma(G)$ is the minimum number of edges one has to remove from $G$
                          such that the resulting graph is acyclic. 
	\end{itemize}
\end{lemma}

\begin{proof}
Assume a subgraph $G'$ of $G$ (with all $n$ nodes) has $m'>0$ edges.
If we remove one edge $e'$ from $G'$ to obtain $G''$, we have, using (\ref{eq:500}) twice:
$$
\gamma(G'') = (m'-1) - n + \zeta(G'') = \gamma(G') - (1 - (\zeta(G'')-\zeta(G'))). 
$$
We observe: 
\begin{itemize}
	\item If $e'$ is a cycle edge in $G'$, then $\zeta(G'')=\zeta(G')$, and hence $\gamma(G'') =\gamma(G')-1$. 
\item If $e'$ is not a cycle edge, then $\zeta(G'') = \zeta(G') + 1$, and hence $\gamma(G'') = \gamma(G')$. 
\end{itemize}
We prove (a): By what we just observed,
to reduce the cyclomatic number from $\gamma(G)$ to 0 
the number of rounds in which a cycle edge is removed must be $\gamma(G)$.
Now we prove (b): Think of the edges as being removed sequentially. Again, by our observation, 
in order to reduce the cyclomatic number from $\gamma(G)$ to 0 by removing as few edges as possible
we should never remove an edge that is not on a cycle. In this way we remove exactly $\gamma(G)$ (cycle) edges. 
\qed
\end{proof}
We have defined the excess $\text{ex}(G)$ of a graph $G$ as the minimum number of edges 
one has to remove from $G$ so that the remaining subgraph has only acyclic
and unicyclic components. In \cite{Kutzelnigg10} the characterization of this quantity 
given next was used as a definition; the same idea was used in~\cite{stash:journal:09} (without giving it a name).  

For $G$ a graph, let $\zeta_{\text{cyc}}(G)$ denote the number of cyclic components of $G$.

\begin{lemma}\label{lem:excess:one}
In all graphs $G$ the equation
\emph{$\text{ex}(G) = \gamma(G) - \zeta_{\text{cyc}}(G)$} is satisfied.
\end{lemma}
\begin{proof} Assume $G$ has $n$ nodes and $m$ edges.\\
``$\le$'': Starting with $G$, we iteratively remove \emph{cycle}
edges until each cyclic component has only one cycle left.
The number of edges removed is at least $\text{ex}(G)$. 
Call the resulting graph $G'$.
Removing one cycle edge from each of the $\zeta_{\text{cyc}}(G)$ cyclic components of $G'$
will yield an acyclic graph. 
Lemma~\ref{lem:cyclo:one}(a) tells us that together exactly $\gamma(G)$ edges have been removed;
hence $\gamma(G) \ge \text{ex}(G) + \zeta_{\text{cyc}}(G)$.\\
``$\ge$'': 
Choose a set $E^+$ of $\text{ex}(G)$ edges in $G$ such that removing these edges leaves
a graph $G'$ with only acyclic and unicyclic components.
Now imagine that the edges in $E^+$ are removed one by one in an arbitrary order.
Let $\beta$ denote the number of edges in $E^+$ that are on a cycle when removed;
the other $\text{ex}(G)-\beta$ many were non-cycle edges when removed.
Removing one cycle edge from each cyclic component of $G'$ 
will leave an acyclic graph.
Counting the number of cycle edges we removed altogether, and applying 
Lemma~\ref{lem:cyclo:one}(a) again, we see that 
$\gamma(G) = \beta + \zeta_{\text{cyc}}(G')$.
Since removing a non-cycle edge from a graph can increase the number of cyclic components by at most 1, 
we have that $\zeta_{\text{cyc}}(G') \le \zeta_{\text{cyc}}(G) + (\text{ex}(G)-\beta)$.
Combining the inequalities yields $\gamma(G) \le \zeta_{\text{cyc}}(G) + \text{ex}(G)$.  
\qed
\end{proof}

\subsection{The Excess of the Cuckoo Graph and the Stash Size}\label{subsec:stash:excess}
The purpose of this section is to prove Lemma~\ref{lemma:stash:excess},
which we recall here.
We assume that $h_1$ and $h_2$ are given,
and write $G(S)$ for $G(S,h_1,h_2)$, for $S\subseteq U$.
\newcounter{tmp}
\setcounter{tmp}{\arabic{lemma}}
\setcounter{lemma}{0}
\begin{lemma}[\cite{stash:journal:09}]\label{lemma:stash:excess:repeated}
The keys from $S$ can be stored in the two tables and a stash of size $s$ using $(h_1,h_2)$ if and only
if \emph{$\text{ex}(G(S))\le s$}. 
\end{lemma}
\setcounter{lemma}{\arabic{tmp}}
\begin{proof}
``$\Rightarrow$'': Assume $T$ is a subset of $S$ of size at most $s$
such that all keys from $S'=S-T$ can be stored in the two tables.
Then all components of $G(S')$ must be acyclic or unicyclic.
(Assume $C$ is a component with $\gamma(C)>1$.
Then by (\ref{eq:500}) the number of edges (keys) in $C$  
would be strictly larger than the number of nodes (table positions), 
which is impossible.)
Since $G(S')$ is obtained from $G(S)$ 
by removing the edges $(h_1(x),h_2(x))$, $x\in T$, 
we get $\text{ex}(G(S))\le s$.\\
``$\Leftarrow$'': Assume $\text{ex}(G(S))\le s$. 
Choose a subset $T$ of $S$ of size $\text{ex}(G(S))$
such that $G(S-T)$ has only acyclic and unicyclic components. 
From what is known about the behaviour of standard cuckoo hashing, 
we can store $S'=S-T$ in the two tables using $h_1$ and $h_2$ (e.\,g., see~\cite[Sect.~4]{devroye}).
(This can even be proved directly. If one of the nodes touched by an edge $(h_1(x),h_2(x))$, $x\in S'$,
has degree 1, we place $x$ in the corresponding cell. Iterating this,
we can place all keys excepting those that belong to cycle edges. 
Since $G(S')$ has only acyclic and unicyclic components,
the cycle edges form isolated simple cycles, and clearly the keys that belong to 
such a cycle can be placed in the corresponding cells.)
By assumption, the keys from $T$ fit into the stash.
\qed 
\end{proof}

\subsection{The Insertion Procedure}\label{subsec:insert}
We consider here the obvious generalization of the insertion procedure 
in standard cuckoo hashing~\cite{cuckoo_hashing_pagh}.
It assumes that a procedure \textit{rehash} is given that will choose
two new hash functions and insert all keys anew.
The parameter \textit{maxloop} is used for avoiding infinite loops.
(When using cuckoo hashing with a stash of size $s$, one will choose $\textit{maxloop}=\rmTheta((s+2) \log n)$.
For analysis purposes, larger values of \emph{maxloop} are considered as well.)
Empty table cells contain \textbf{nil}. The operation \textit{swap}
exchanges the contents of two variables. 
\begin{algorithm}[Insertion in a cuckoo table with a stash]\samepage\label{algo:insert}
\begin{tabbing}
		\qquad\=\quad\=\quad\=\quad\=\quad\=\quad\=\quad\=\quad\=\quad\=\kill\\[-3.6ex]
		\textbf{procedure} \textit{stashInsert}($x$: \texttt{key})\\
		{\small(1)} \>\> $\texttt{nestless} := x$;\\
		{\small(2)} \>\> $\texttt{i} := 1$;\\
		{\small(3)} \>\> \textbf{repeat} \textit{maxloop} \textbf{times}\\
		{\small(4)} \>\>\>\> \textit{swap}$(\texttt{nestless},T_{\texttt{i}}[h_{\texttt{i}}(\texttt{nestless})])$;\\ 
		{\small(5)} \>\>\>\> \textbf{if} $\texttt{nestless}=\textbf{nil}$ \textbf{then} \textbf{return};\\
		{\small(6)} \>\>\>\> $\texttt{i} := 3-\texttt{i}$;\\ 
		{\small(7)} \>\> \textbf{if} stash is not yet full\\
		{\small(8)} \>\>\>\> \textbf{then} add \texttt{nestless} to stash\\
		{\small(9)} \>\>\>\> \textbf{else} \textit{rehash}.
\end{tabbing}
\end{algorithm}
As long as it is not finished, the procedure maintains a ``nestless'' key (in \texttt{nestless})
and the current index $i\in\{1,2\}$ (in \texttt{i}) of the table where this key is to be placed.
When a new key $x$ is to be inserted, it is declared ``nestless'' and $i$ is set to 1. 
As long as there is a nestless key $x$, but at most
for \textit{maxloop} rounds, the following is iterated: 
Assume $x$ is nestless and the current index is $i$. 
Then $x$ is placed in position $h_i(x)$ in table $T_i$.
If this position is empty, the procedure terminates;
if it contains a key $x'$, 
that key gets evicted to make room for $x$,
is declared nestless, and $i$ is changed to the other value $3-i$. 
If the loop does not terminate within \textit{maxloop} rounds, the
key that is currently nestless gets stored in the stash.
If this causes the stash to overflow, a \emph{rehash} is carried out. 
(This may be realized by 
collecting all keys from tables and stash as well as the nestless key,
choosing a new pair $(h_1,h_2)$ of hash functions,  
and calling the insertion procedure for all keys.)

\subsection{Complete Insertion Loops and the Excess}
We first look at the behavior of certain variants of the insertion procedure
(called ``complete''), 
which exhibit the following behavior when $x$ is inserted:
(i) if with \emph{maxloop} set to infinity the loop were to run forever,
then this is noticed and at some point the currently nestless key is put in the stash; 
(ii) otherwise the loop is left to run until the nestless key is stored in
an empty cell.
It is not hard to see (\emph{cf.}~\cite{devroye}) that one obtains a complete variant from Algorithm~\ref{algo:insert} 
if one chooses \textit{maxloop} as some number larger than $2|S|+3$.

\begin{proposition}[\cite{stash:journal:09,Kutzelnigg10}]\label{prop:complete:insertion}
If inserting the keys of $S$ by some complete insertion procedure places $s$ keys in the stash, 
then \emph{$s=\text{ex}(G(S))=\text{ex}(G(S,h_1,h_2))$}. 
\end{proposition}
\begin{proof}
``$\ge$'': After the insertion is complete, all keys from $S$ are stored in the two tables and the stash.
Lemma~\ref{lemma:stash:excess:repeated} implies that $s\ge\text{ex}(G(S))$.\\
``$\le$'': For this, we use induction on the size of $S$. If $S=\emptyset$,
excess and stash size are both 0. Now assume as induction hypothesis that set $S$ has been inserted,
that the set of keys placed in the stash is $T$, and that $|T|=s\le\text{ex}(G(S))$.
Let $S'=S-T$.   
We insert a new key $y$ from $U-S$.\\ 
\emph{Case} 1: The insertion procedure finds that $y$ can be accommodated
without using the stash.---The stash size remains $s$, and  
$s \le \text{ex}(G(S))\le \text{ex}(G(S\cup\{y\}))$.\\
\emph{Case} 2: The complete insertion procedure notices that the loop were to run forever
 and places some key in the stash.---By the properties of the complete insertion loop for standard cuckoo hashing 
 as explored in~\cite{devroye} we know that $G(S'\cup\{y\})$ must contain a connected component
 that is neither acyclic nor unicyclic. 
 Since $\text{ex}(G(S'))=0$, it must be edge $(h_1(x),h_2(x))$ that makes the difference.
This means that each endpoint of $(h_1(x),h_2(x))$ lies in some cyclic component of $\text{ex}(G(S'))$.
Now $G(S')$ is a subgraph of $G(S)$, 
so the same is true in $G(S)$.
Recall Lemma~\ref{lem:excess:one}, and consider two cases when changing from $S$ to $S\cup\{y\}$:
If the endpoints of $(h_1(x),h_2(x))$ lie in two different cyclic components of $G(S)$,
then the number of cyclic components decreases by 1, hence the excess increases by 1;
if they lie in one and the same cyclic component, 
then the cyclomatic number increases by 1, and the excess increases by 1 as well. 
In both cases we get that $s+1 \le \text{ex}(G(S)) +1  = \text{ex}(G(S\cup\{y\}))$. 
\qed
\end{proof}

\subsection{Standard Insertion}
It turns out that by choosing $\textit{maxloop}=\rmTheta((s+2)\log n)$ in Algorithm~\ref{algo:insert} 
we can make sure that with probability of $O(1/n^{s+1})$ no rehash is necessary.%
\footnote{If deletions are allowed, before calling \emph{rehash}
one should try whether any one of the $s$ keys presently stored in the stash
can be inserted into the tables by the insertion procedure. We ignore deletions here.} 
Note that if the stash has size $0$, then Algorithm~\ref{algo:insert} is
exactly the insertion procedure of standard cuckoo hashing
from~\cite{cuckoo_hashing_pagh}.
The following claim and proof are similar 
to what has to be done in the analysis of standard cuckoo hashing. 

\begin{proposition}\label{app:prop:random}
Assume the hash functions $(h_1^*,h_2^*)$ are fully random, and 
the keys from $S$ are inserted sequentially into a cuckoo table with a stash of size $s$,
using Algorithm~\ref{algo:insert}.
If we choose $\textit{maxloop}=\alpha (s+2)\log n$ for a suitable constant $\alpha>0$, 
then we have\emph{:}\emph{
$$
\Pr(\text{the stash of size $s$ overflows})=O(1/n^{s+1}).
$$
}\end{proposition}
\emph{Sketch of proof.} Theorem~\ref{thm:r_excess_prob_zz}
tells us that the probability that a stash of size $s$ is not sufficient 
because the excess of $G(S,h_1^*,h_2^*)$ is too large is $O(1/n^{s+1})$.
All we have to show is that the probability is also this small that an extra key slips into the
stash because the insertion loop was stopped by the step counter hitting \emph{maxloop}. 
Let $S=\{x_1,\ldots,x_n\}$, with the keys listed in the order in
which they are inserted, and let $S_j=\{x_1,\ldots,x_j\}$.  
For $1\le j \le n$ and some bound $p$ consider the event that the insertion procedure for $x_j$
needs $p$ or more rounds. 
One can show (this was done in~\cite{cuckoo_hashing_pagh} with a different terminology)
that then $G(S_j)$ must contain a path $u_0,u_1,\ldots,u_t,u_{t+1}$, with $t=\lceil p/3\rceil$,
where $u_0$ is the node corresponding to $T_i[h_i(x_j)]$, for $i=1$ or $i=2$,
and $u_0,\ldots,u_t$ are distinct nodes. 
Viewing the situation in terms of edges this means that 
there must be some $T\subseteq S_{j-1}$ of size $t=\lceil p/3\rceil$
such that the edges $(h_1(x),h_2(x))$, with $x\in T$, in some order, form such a path.
The latter event we call $\mathcal{A}_{T}$. 
We have $\Pr(\mathcal{A}_{T}) \le t! \cdot 2/m^t$.
The number of sets $T$ to consider is $\binom{j-1}{t} < \binom{n}{t}$.
Thus, 
\begin{equation}
\Pr(\exists T\subseteq S_{j-1}\colon |T|=t\wedge \mathcal{A}_{T}) \le \binom{n}{t}t!\cdot \frac{2}{m^t} < \frac{2}{(1+\varepsilon)^t}.
\label{eq:10001}
\end{equation}
If the insertion of $x_j$ increases the stash size although $\text{ex}(G(S_j))=\text{ex}(G(S_{j-1}))$,
then this insertion must make $p=\textit{maxloop}$ steps. 
By (\ref{eq:10001}), the probability of this to happen is smaller than $2/{(1+\varepsilon)^{\lceil \textit{maxloop}/3\rceil}}$.
So, if we choose $\textit{maxloop} = 3(s+2)\log_{1+\varepsilon} n$ (i.\,e., $\alpha=3/\log(1+\varepsilon)=\rmTheta(1/\varepsilon)$), this 
probability will be smaller than $1/n^{s+2}$. 
Summing over all $j$ we obtain the bound $O(1/n^{s+1})$.\qed

\begin{proposition}\label{app:prop:insertion}
Assume the hash functions $(h_1^*,h_2^*)$ are fully random, the
keys from $S$ are stored in a cuckoo table with a stash of size $s$,
and a new key $y$ is inserted by Algorithm~\ref{algo:insert},
with $\textit{maxloop}\ge\alpha (s+2)\log n$ 
for $\alpha$ as in Proposition~\ref{app:prop:random}. 
Then the expected number of steps needed for this insertion is $O(1)$. 
\end{proposition}
\emph{Sketch of proof.} 
Let the random variable $Z$ denote the number of rounds needed for this insertion.
We ignore the cost of a \emph{rehash}.
(The contribution of this rare event to the overall insertion time is $O(1/n^s)$.
A discussion of the case $s=0$ can be found in~\cite{DW2003a}.)
Then $\E(Z)\le \sum_{p\ge 1}\Pr(\text{at least $p$ rounds are needed to store $y$})$, and hence, 
using (\ref{eq:10001}) and arguing as in~\cite{cuckoo_hashing_pagh}, 
\begin{equation*}
\E(Z)\le \sum_{p\ge 1}\Pr(\exists T\colon |T| = \lceil p/3\rceil \wedge \mathcal{A}_T) \le 
\sum_{p\ge 1} \frac{2}{(1+\varepsilon)^{\lceil p/3\rceil}} = O(1).
\tag*{\qed}
\end{equation*}
Of course, the last two propositions are formulated for fully random hash functions. 
Using the techniques developed for the proof of Theorem~\ref{thm:r_excess_prob_zz}
one can show that they are valid for hash functions from $\ZZ^{2k,c}_{\ell,m}$ as well,
for the parameter choices as in that theorem. 
The only difference is that instead of the graph property $\LL$ one has to use the 
graph property ``connected'' in a way explored in detail in~\cite{DW2003a}.

\section{Proof of a Claim}\label{app:claim:proof}
We prove the following claim, stated and used in the proof of Lemma~\ref{lem:02}.
\begin{claim}
    Every leaf\/less connected graph with $i$ marked edges has a leaf\/less connected subgraph with 
    cyclomatic number $\le i {+} 1$ that contains all marked edges. 
\end{claim}

\begin{proof}
Let $G=(V,E)$ be a leaf\/less connected graph. If $\gamma(G) \le i+1$, there is nothing to prove. 
Thus assume $\gamma(G) \ge i+2$. Choose an arbitrary spanning
tree $(V,E_0)$ of $G$.

There are two types of edges in $G$:
\emph{bridge edges} and \emph{cycle edges}. A bridge edge is an edge whose
deletion disconnects the graph,
cycle edges are those whose deletion does not disconnect the graph.

Clearly, all bridge edges are in $E_0$.  
Let $E_{\text{mb}}\subseteq E_0$ denote the set of marked bridge edges. 
Removing the edges of $E_{\text{mb}}$ from $G$ will split $V$ into $|E_{\text{mb}}|+1$ connected components $V_1,\ldots,V_{|E_\text{mb}|+1}$;
removing the edges of $E_{\text{mb}}$ from the spanning tree $(V,E_0)$ will give exactly the same components.
For each \emph{cyclic} component $V_j$ we choose one edge $e_j\notin E_0$ that connects two nodes in $V_j$.
The set of these $|E_\text{mb}|+1$ edges is called $E_1$. Now each marked
bridge edge lies on a path connecting two cycles in $(V,E_0 \cup E_1)$.

Recall from graph theory~\cite{Diestel} the notion of a fundamental cycle:
Clearly, each edge $e\in E-E_0$ closes a unique cycle with $E_0$.
The cycles thus obtained are called the fundamental cycles of $G$ w.\,r.\,t. the spanning tree $(V,E_0)$.
Each cycle in $G$ can be obtained as an XOR-combination of fundamental cycles.
(This is just another formulation of the standard fact that the fundamental cycles form a basis
of the ``cycle space'' of $G$, see~\cite{Diestel}.)
From this it is immediate that every cycle edge of $G$ lies on some fundamental cycle. 
Now we associate an edge $e'\notin E_0$ with each marked cycle edge $e\in E_{\text{mc}}$.
Given $e$, let
$e'\notin E_0$ be such that $e$ is on the fundamental cycle of $e'$.
Let $E_2$ be the set of all edges $e'$ chosen in this way. Clearly, each $e \in
E_{\text{mc}}$ is a cycle edge in $(V,E_0\cup E_2)$.

Now let $G' =(V,E_0 \cup E_1 \cup E_2)$. Note that
$|E_1\cup E_2| \leq(|E_\text{mb}| + 1) +|E_\text{mc}| \le i+1$ and thus $\gamma(G') \leq i + 1$. 
In $G'$, each marked edge is on a cycle or on a path that connects two cycles. 
If we iteratively remove leaf edges from $G'$ until no leaf is left, none of the marked edges will be affected. 
In this way we obtain the desired leaf\/less subgraph $G^\ast$ with $\gamma(G^\ast)=\gamma(G')\le i+1$.
\qed
\end{proof}

\end{document}